\documentclass[a4paper,11pt]{amsart}
\usepackage{graphicx}
\usepackage{amsmath,amssymb}
\newtheorem{theorem}{Theorem}[section]

\newtheorem{lemma}[theorem]{Lemma}
\newtheorem{proposition}[theorem]{Proposition}
\theoremstyle{definition}
\newtheorem{definition}[theorem]{Definition}
\theoremstyle{remark}

\numberwithin{equation}{section}

\begin{document}

\title{A Phase Transition in a Quenched Amorphous Ferromagnet}%
\author{Alexei Daletskii}
\address{Department of Mathematics,
University of York,
York YO1 5DD, UK}
\email{alex.daletskii@york.ac.uk}
\author{Yuri Kondratiev}%
\address{Fakult\"at f\"ur Mathematik, Universit\"at Bielefeld, D-33501 Bielefeld}%
\email{kondrat@math.uni-bielefeld.de}%
\author{Yuri Kozitsky}
\address{Instytut Matematyki, Uniwersytet Marii Curie-Sklodowskiej, 20-031 Lublin, Poland}
\email{jkozi@hektor.umcs.lublin.pl}
\author{Tanja Pasurek}
\address{Fakult\"at f\"ur Mathematik, Universit\"at Bielefeld, D-33501 Bielefeld}%
\email{tpasurek@math.uni-bielefeld.de}



\maketitle

\begin{abstract}
Quenched thermodynamic states of an amorphous ferromagnet are
studied. The magnet is a countable collection of point particles
chaotically distributed over $\mathbb{R}^d$, $d\geq 2$. Each
particle bears a real-valued spin with symmetric a priori
distribution; the spin-spin interaction is pair-wise and attractive.
Two spins are supposed to interact if they are neighbors in the
graph defined by a homogeneous Poisson point process. For this model,
we prove that with probability one: (a) quenched thermodynamic
states exist; (b) they are multiple if the particle density (i.e.,
the intensity of the underlying point process) and the inverse
temperature are big enough; (c) there exist multiple quenched thermodynamic
states which depend on the realizations of the underlying point process in a measurable way.

\end{abstract}

\section{Introduction}
\subsection{Setup}

In this paper, we study thermodynamic states of the following system.
 A countable collection of
`particles' is distributed over $\mathbb{R}^d$, $d\geq 2$, in such a
way that every bounded $\Lambda \subset\mathbb{R}^d$ contains only
finite number of them. Each `particle' represents a cluster of
negligible size of magnetically active physical particles. In our
model, this amounts to assuming that a `particle' is characterized
by its location $x\in \mathbb{R}^d$ and spin $\sigma_x\in
\mathbb{R}$. The locations of `particles' constitute a locally finite
set (configuration) $\gamma \subset \mathbb{R}^d$, and the spins  take real
values. We also assume that each $\sigma_x$ is characterized by one
and the same symmetric a priori distribution $\chi$ on $\mathbb{R}$.
The interaction is supposed to be pair-wise and attractive. For the
`particles' located at $x$ and $y$, the interaction energy  is
$\phi(|x-y|) \sigma_x \sigma_y$, with $\phi(r)\geq \phi_* >0$ for
$r\in [0, r_*]$, and $\phi(r) =0$ for $r >r_*$. If $\gamma$
were a crystalline lattice, then the model would be a standard
lattice system of `unbounded spins'. The study of Gibbs states of such
spin systems goes back to the seminal paper \cite{LP}, further continued in
\cite{Bell,Park}. In \cite{KP}, a similar model living on a more general
discrete metric space was studied. The next step was made in \cite{KKP} where
the underlying set was a countable graph with globally unbounded  vertex degrees.
In that paper, a class of graphs was introduced in which vertices of large degree are sparse -- a property
formulated in \cite{KKP} as a weighted summability of the vertex degrees. For such graphs, tempered Gibbs states
of unbounded spin systems were constructed and studied. A natural continuation of those works
would be to pass to random graphs, in which this kind of  summability holds for almost all realizations.
In the present paper we do this step by letting
the underlying set $\gamma$ be random, obeying the Poisson law with
homogeneous density $\lambda >0$. The graph structure on $\gamma$ is then
defined by the spin-spin interaction: $x, y\in \gamma$ are adjacent (neighbors) if $\phi(|x-y|) >0$.
 We call this model the {\it
amorphous ferromagnet}, cf. \cite[Section 11]{OHandley}.

In view of the mentioned randomness, there are two types of
thermodynamic state of our model. In the first case, the randomness is taken into account
already at the level of local states defined on the space of
(marked) configurations $\hat{\gamma} = \{ (x, \sigma_x): x\in
\gamma\}$. The global thermodynamic states constructed in this way
are then {\it annealed states}; they describe the thermal equilibrium of the
whole system. The second approach, which we follow in this paper,
consists in constructing thermodynamic states of the spin system
alone for fixed {\it typical} configurations $\gamma$. These are
{\it quenched states}, cf. \cite{BA}. The global observables characterizing such
states ought to be {\it self-averaging} --  taking the same value
for all typical (i.e., for  almost all) configurations $\gamma$.
Note that studying quenched states is a more difficult problem, as
compared to that of annealed ones, since the present spatial
irregularities do not allow for applying here most of the methods
effective for regular systems. In what follows, we aim at proving:
\begin{itemize}
  \item Existence, for  almost all configurations $\gamma$, of thermodynamic
  states with properties suitable for physical applications.
  \item Measurability of thermodynamic states with respect to $\gamma$.
  \item Multiplicity of such states, for  almost $\gamma$,
  for temperatures  $T< T_*$, where the self-averaging parameter $T_*>0$ may
  depend  on the model parameters $\lambda$, $\phi_*$,  $r_*$, $d$.
\end{itemize}
There are only few publications on the mathematically
rigorous theory of phase transitions in spin systems of general type
living on non-crystalline (amorphous) substances, cf.
\cite{ChayesZ,GZ,GruberG,Gruber,Romano} where annealed states were
considered. The reason for this is presumably that the corresponding methods,
e.g., infrared estimates, are essentially
based on the translation invariance (and other symmetries) of the
underlying crystals. At the same time, for Ising spins $\sigma_x =
\pm 1$, there exist methods applicable to the corresponding models
on graphs, cf. \cite{H,Lyons}. For such models, see also
\cite{Chayes}, the main idea of proving the existence of phase
transitions is to relate the appearance of multiple phases of the
spin system to the Bernoulli bond percolation on the underlying
graph. In our model, we deal with a random graph with vertex set
$\gamma$ and the adjacency relation $x \sim y$ defined by the
property $|x-y|\leq r_*$. That is, the set of edges of the graph is $\varepsilon_\gamma=\{ \{x,y\}
\subset \gamma: x \sim y\}$, and the graph itself -- known as the {\it Gilbert graph} -- is then the pair $(\gamma,
\varepsilon_\gamma)$. It
has various applications and is intensively studied, see, e.g., \cite{Balister,GeL,MR,Penrose}.
The probability distribution of $(\gamma,
\varepsilon_\gamma)$ is described in
subsection \ref{ssec.2.2} below. If this graph has an infinite
connected component, which is a random event with probability
dependent on $\lambda$, under certain conditions one can observe
the Bernoulli bond percolation with a nontrivial percolation
threshold $q_*\in (0,1)$. We combine the mentioned methods and prove
that the mean magnetization in the model can be positive for almost
all configurations $\gamma$, and hence the quenched Gibbs states can
be multiple, if the particle density $\lambda$ and the inverse
temperature $\beta = 1/T$  are large enough\footnote{From now on, the parameters
$r_*$ and $\phi_*$  are fixed and mostly suppressed from
the notations.}.  Finally, let us mention that, for our model with
$\sigma_x \in \mathbb{R}$, the problem of uniqueness of Gibbs states
remains open, see subsection
\ref{2.4ss} below.  We also note
that the method developed in this article can be used to study
annealed states of amorphous ferromagnetic substances where the spin and the particle configurations are in thermal equilibrium.

\subsection{The overview of the results}

In  the sequel, by $\pi_\lambda$ we denote the Poisson measure with
density $\lambda >0$ which describes the probability distribution of
the configurations $\gamma$. In Proposition \ref{1pn} and Theorem \ref{1tm} below, we show
that there exists a set of configurations $A_1$ such that: (a) $\pi_\lambda (A_1) =1$ for all
$\lambda>0$; (b) $\mathcal{G}_{\rm t} (\beta|\gamma) \neq \emptyset$ for all
$\beta >0 $ and all $\gamma\in A_1$. Here $\mathcal{G}_{\rm t}
(\beta|\gamma)$ is the set of {\it tempered} Gibbs measures
of our spin system on $\gamma$  at a given $\beta$.
Tempered Gibbs measures are Gibbs measures supported on
the configurations with tempered growth of $|\sigma_x|$ as $|x| \to
+\infty$. Note that, for spin models on graphs with unbounded vertex
degrees and with single-spin distributions with noncompact support,
there may exist states supported on configurations of spins with
rapidly increasing $|\sigma_x|$, whereas for typical ferromagnetic
configurations in physical substances, most of the spins take values
close to same $s>0$. The proof of Theorem \ref{1tm} is based on
\cite[Theorem 3.1]{KKP} and on the property of $\pi_\lambda$
obtained in Proposition \ref{1pn}, in which $A_1$ appears as the set
of all those configurations for which the quantities in (\ref{9}) are finite, and hence
the graph $(\gamma, \varepsilon_\gamma)$ belongs to the class of sparse graphs introduced in \cite{KKP}.

The Ising model on $\gamma$
is a particular case of our model, which corresponds to the choice
$\chi(d \sigma) = \delta(\sigma^2 - 1) d \sigma$. As is well-known, the set
of the Gibbs measures of this model, $\mathcal{G}^{\rm Ising}_{\rm
t} (\beta|\gamma)$, is nonempty for all $\gamma$. Next, for
\begin{eqnarray*}
A_2 (\beta) & := &\{ \gamma : |\mathcal{G}_{\rm t} (\beta|\gamma)|>1\}, \\[.2cm]
A^{\rm Ising}_2 (\beta) & := &\{ \gamma : |\mathcal{G}^{\rm
Ising}_{\rm t} (\beta|\gamma)|>1\}, \nonumber
\end{eqnarray*}
by the Wells correlation inequality \cite{W}, it follows that, cf.
Proposition \ref{Wpn} below,
\begin{equation}
  \label{A3}
  A_2 (\beta) \supset  A^{\rm Ising}_2 (a^2\beta)
\end{equation}
where $a>0$ is determined by the measure $\chi$, see (\ref{21}). For
the reader convenience, we present here a complete proof of the Wells
inequality, which is a refinement of that in
\cite[Appendix]{Bricmont}. By standard results on the continuum
percolation driven by the Poisson random point process, see
\cite{MR,Penrose} and also \cite[Corollary 3.7]{GH} and
\cite[Theorem 3.1]{GeL}, it follows that, for $\pi_\lambda$-almost
all $\gamma$, the corresponding graph $(\gamma, \varepsilon_\gamma)$
has an infinite connected component whenever $\lambda > \lambda_*$,
where a non-random parameter $\lambda_*>0$ is determined by the
parameter $r_*$ and the dimension of the space $d$. Suppose now that
each edge of this infinite connected component is removed
independently with probability $1-q$ and kept with probability $q$.
If the graph obtained in this way still possesses an infinite connected
component, then one says that the Bernoulli bond percolation with
bond probability $q$ occurs on the infinite connected
component of $(\gamma,\varepsilon_\gamma)$. As in the
previous result, it is possible to show, see Propositions \ref{GHpn}
and \ref{GHpn1} below, that for $\lambda > \lambda_*$ there exists
$q_*\in (0,1)$ such that
\begin{equation}
  \label{A4}
 \pi_\lambda (A_3 (q)) =1, \qquad {\rm for} \ \ q> q_*  \ \ {\rm and} \ \  \lambda > \lambda_*,
\end{equation}
where $A_3 (q)$ is the set of $\gamma$ such that the Bernoulli bond percolation with bond probability
$q$ occurs on the infinite connected component of
$(\gamma,\varepsilon_\gamma)$. By \cite[Theorem 2.1]{H}, we know
that
\begin{equation*}
A^{\rm Ising}_2 (\beta) \supset A_3 (q), \quad {\rm for} \ \ \beta > [\log (1+ q) - \log (1- q)]/2,
\end{equation*}
Then combining (\ref{A3}) and (\ref{A4}), we conclude that, for
$\lambda> \lambda_*$, there exists $\beta_* \geq [\log (1+ q_*) -
\log (1- q_*)]/2$ such that,
for all $\beta > a^2 \beta_*$ and $\pi_\lambda$-almost all $\gamma$, the set
$\mathcal{G}_{\rm t} (\beta|\gamma)$ contains at least two elements, see Theorem \ref{2tm} below.
In principle, we could stop here. However, in that case one important aspect of the theory would have been omitted.
This is the dependence of our tempered Gibbs measures on $\gamma$.
In mathematical theories of random systems \cite{Bov}, Gibbs measures  are supposed to
depend on the random parameters in a measurable way. Then they are called {\it random Gibbs measures}.
 In Section \ref{Quensec}, we study the measurability issue by employing marked configurations $\hat{\gamma}$ consisting of pairs
$(x, \sigma_x)$. In this setting, random Gibbs measures are obtained as  conditional  measures on the space of marked configurations.
In Theorem \ref{etatm}, we show that  random Gibbs measures of our model are multiple if the conditions of Theorem \ref{2tm} are satisfied.

For the sake of clarity, in this paper we restricted
ourselves to the simplest model of amorphous substances -- the
Gilbert graph model based on the Poisson point process. In a similar way, one can prove the statements
mentioned above if the underlying graph is as in  the random
connection model, see \cite{GeL,MR,Penrose} or a tempered Gibbs
random field, see \cite[Corollary 3.7]{GH}. The only condition is that the graph almost surely has the
summability property as in Proposition \ref{1pn}, see \cite{DKKP} for more detail.

\section{Quenched Gibbs measures}
\label{sec:2}

\subsection{The underlying graph}

\label{ssec.2.2}

Let $\mathcal{B} (\mathbb{R}^d)$ and
$\mathcal{B}_{\rm b} (\mathbb{R}^d)$ stand for the set of all
Borel and all bounded Borel subsets of $\mathbb{R}^d$,
respectively.
The space of all configurations is defined as
\begin{equation}
\label{C1}
 \Gamma=\{\gamma\subset\mathbb R^d :
 |\gamma\cap \Lambda|<\infty \quad {\rm for}{\rm \ any} \  {\rm compact} \ \Lambda\subset \mathbb{R}^d
\},
\end{equation}
where $|A|$ stands for the cardinality of $A$. The
space $\Gamma$ is equipped with the vague topology being the weakest
one in which the maps $\Gamma \ni \gamma \mapsto \sum_{x\in \gamma}
f(x)$ are continuous for all continuous functions $f:\mathbb{R}^d
\to \mathbb{R}$ with compact support, see, e.g., \cite{Albev}.
By $\mathcal{B}(\Gamma)$ we denote  the corresponding
Borel $\sigma$-field. The vague topology is metrizable in such a way
that  the corresponding metric space is complete and separable. For
$\lambda >0$, by $\pi_\lambda$ we denote the homogeneous Poisson
measure on $(\Gamma, \mathcal{B}(\Gamma))$ with intensity (density)
$\lambda$. Note that the set of all finite configurations $\Gamma_0$
is a Borel subset of $\Gamma$ such that $\pi_\lambda (\Gamma_0)=0$.
That is, $\pi_\lambda$-almost all configurations $\gamma \in \Gamma$
are infinite.

Each $\gamma$ can be considered as
a graph with vertex set $\gamma$ and the adjacency relation $x \sim
y$ defined by the condition $|x-y|\leq r_*$.  Then
$\varepsilon_\gamma
= \{ \{x,y\} \subset \gamma: x \sim y\}$ is its edge set. The
probability distribution of the random graph
$(\gamma,\varepsilon_\gamma)$ is constructed in the following way,
cf. \cite{GH}. Let $\varepsilon$ denote a set of pairs of distinct
points, i.e., of $e=\{x,y\}$, $x, y \in \mathbb{R}^d$, $x\neq y$. We
say that $\varepsilon$ is locally finite if $\varepsilon_\Lambda$ is
finite for each $\Lambda \in\mathcal{B}_{\rm b}(\mathbb{R}^d)$. Here
$\varepsilon_\Lambda := \{ \{x,y\} \in \varepsilon : \{x,y\} \subset
\Lambda \}$. Let $E$ be the set of all locally finite $\varepsilon$,
and $\mathcal{F}(E)$ be the $\sigma$-field of subsets of $E$
generated by the counting maps $\varepsilon \mapsto
|\varepsilon_\Lambda|$ with all possible $\Lambda\in\mathcal{B}_{\rm
b}(\mathbb{R}^d)$. Note that $E$ is, in fact, the set of locally
finite configurations, cf. (\ref{C1}), over the symmetrization of
the set $\mathbb{R}^d \times \mathbb{R}^d \setminus \{ (x,x): x \in
\mathbb{R}^d\}$. Let $\mathcal{P}(E)$ be the set of all probability
measures on $(E, \mathcal{F}(E))$. Each $\varsigma \in
\mathcal{P}(E)$ can uniquely be determined by its Laplace transform,
which we introduce in the following way. Let  $\theta: \mathbb{R}^d
\times \mathbb{R}^d \to (-1,0]$ be measurable, symmetric, and local,
i.e., $\theta(x,y) = 0$ whenever $x$ or $y$ is in $\Lambda^c:=
\mathbb{R}\setminus \Lambda$ for some $\Lambda \in \mathcal{B}_{\rm
b}(\mathbb{R}^d)$. By $\Theta$ we denote the set of all such
functions. Then for each $\varsigma \in \mathcal{P}(E)$, its Laplace
transform is defined as
\begin{equation*}
L_\varsigma (\theta) = \int_{E} \exp\left[\sum_{\{x,y\}\in \varepsilon} \log (1 + \theta (x,y))  \right]
\varsigma (d \varepsilon), \qquad \theta \in \Theta.
\end{equation*}
Now let $g: \mathbb{R}^d \times \mathbb{R}^d \to [0,1]$ be
measurable and symmetric. For each $\theta \in \Theta$, the
pointwise product $ g \theta$ is also in  $\Theta$. An {\it
independent  $g$-thinning} of a given $\varsigma\in \mathcal{P}(E)$,
cf. \cite[Section 11.2]{DV}, is the measure $\varsigma^g$ defined by
the relation
\begin{equation}
\label{Laplacee}
L_{\varsigma^g} (\theta) = L_\varsigma (g \theta).
\end{equation}
The $g$-thinning of  $\varsigma$ means that each configuration
$\varepsilon$ distributed according to $\varsigma$ is `thinned' in
the sense that each  $\{x,y\}\in \varepsilon$ is removed from the
edge configuration with probability $1-g(x,y)$ and is kept with
probability $g(x,y)$. The probability distribution of such `thinned'
configurations is then $\varsigma^g$.

Given $\gamma \in \Gamma$, the measure $\varsigma(\cdot |\gamma)\in \mathcal{P}(E)$ is defined by
its Laplace transform
\begin{eqnarray}
  \label{Lapl}
L(\theta|\gamma) =    \exp\left[\sum_{x\in \gamma}\sum_{y\in \gamma\setminus x} \log (1 + j(x,y) \theta (x,y))  \right]
\end{eqnarray}
where $j(x,y)= 1$ for $|x-y|\leq r^*$ and $j(x,y)= 0$ otherwise.
For a measurable $\Psi: \mathbb{R}^d \times \mathbb{R}^d \to
\mathbb{R}_{+}:=[0,+\infty)$, the function
\[
\Gamma \ni \gamma \mapsto \sum_{x \in \gamma} \sum_{y \in
\gamma\setminus x} \Psi(x,y) \in \mathbb{R}
\]
is measurable. Therefore, the map $\Gamma \ni \gamma
\mapsto L (\theta|\gamma)$ is also measurable for each $\theta$.
Thus, $\varsigma $ defined in (\ref{Lapl}) is a probability kernel from
$(\Gamma, \mathcal{B}(\Gamma))$ to $(E, \mathcal{F}(E))$, cf.
\cite[Lemma 2.4]{GH}.
The probability distribution of the graph $(\gamma,\varepsilon_\gamma)$ is now defined by the measure
\begin{equation}
\label{Lapl2}
\zeta_{\lambda} (d \gamma , d \varepsilon) =  \varsigma ( d\varepsilon|\gamma)\pi_\lambda (d\gamma).
\end{equation}
It may happen that typical graphs $(\gamma, \varepsilon_\gamma)$
have only finite connected components. Clearly, the probability of
this event, calculated from (\ref{Lapl2}), depends on $\lambda$.  Given $g$ as in
(\ref{Laplacee}), by $\varsigma^g(\cdot|\gamma)$ we denote the
independent $g$-thinning of $\varsigma(\cdot|\gamma)$, and, cf.
(\ref{Lapl2}),
\begin{equation*}
 \zeta^g_{\lambda} (d \gamma , d \varepsilon) := \varsigma^g ( d\varepsilon|\gamma)\pi_\lambda (d\gamma) .
\end{equation*}
The
following fact is known, see \cite[Lemma 2.4 and Corollary
3.7]{GH} and especially \cite[Theorem 3.1]{GeL}.
\begin{proposition}
  \label{GHpnn}
Let $g: \mathbb{R}^d \times \mathbb{R}^d \to
[0,1]$ be measurable, symmetric, and  translation invariant. Then there exists $c(d)>0$ such that, for $\lambda$ satisfying the condidion
\begin{equation}
  \label{Lapl4}
 \lambda \int_{\mathbb{R}^d} g(x,0) d x > c(d),
\end{equation}
$\zeta^g_{\lambda}$-almost all graphs $(\gamma, \varepsilon)$ have infinite connected components.
\end{proposition}
Then by Proposition \ref{GHpnn} we obtain the following fact.
\begin{proposition}
  \label{GHpn}
There exists $\lambda_*$ such that, for $\lambda > \lambda_*$, the
graph $(\gamma, \varepsilon_\gamma)$ has an infinite connected
component for $\pi_\lambda$-almost all $\gamma$.
\end{proposition}
Indeed, by taking in (\ref{Lapl4})  $g(x,y) = j(x,y)$, we prove Proposition \ref{GHpn} with
\begin{equation}
  \label{Lapl5}
 \lambda_* = c(d)/ V(r^*),
\end{equation}
where $V(r^*)$ is the volume of the ball $\{x\in \mathbb{R}^d: |x| \leq r^*\}$.
Note that there can only be a single infinite
connected component, see \cite[Theorem 6.3, page 172]{MR}.

For a constant function $q(x,y) \equiv q\in [0,1]$, let us consider the independent
$q$-thinning of $\varsigma(\cdot|\gamma)$. If the corresponding
random graph has an infinite connected component, then the Bernoulli
bond percolation with bond probability $q$ occurs on the infinite
connected component of $(\gamma, \varepsilon_\gamma)$.  The next
fact can also be deduced from Proposition \ref{GHpnn}.
\begin{proposition}
  \label{GHpn1}
Let $\lambda_*$ be as in Proposition \ref{GHpn} and inequality
$\lambda >\lambda_*$ hold. Then, there exists $q_* \in (0,1)$ such
that for $\pi_\lambda$-almost all $\gamma$, the Bernoulli bond
percolation with bond probability $q >q_*$ occurs on the infinite
connected component of $(\gamma, \varepsilon_\gamma)$.
\end{proposition}
Indeed, by (\ref{Lapl4}) and  (\ref{Lapl5}) one can take
\begin{equation*}
  q_* = \frac{\lambda_*}{\lambda} = \frac{c(d)}{\lambda V(r^*)}.
\end{equation*}
Our next step is to obtain information on the distribution of the
vertex degrees in $(\gamma, \varepsilon_\gamma)$. In other words, we shall
prove that, for $\pi_\lambda$-almost all $\gamma$, $(\gamma, \varepsilon_\gamma)$ belongs to the class of sparse
graphs studied in \cite{KKP}. For $x\in \gamma$, let $n_\gamma(x)$ be
the number of neighbors of $x$ in $\gamma$, i.e., $n_\gamma (x) :=
|\{y \in \gamma: y \sim x\}|$. Clearly, $n_\gamma (x)$ is finite for
 all $\gamma\in \Gamma$. Note, however, that
\begin{equation}
  \label{A6}
\sup_{x\in \gamma}n_\gamma(x)= +\infty,
\end{equation}
also for $\pi_\lambda$-almost all $\gamma$.
For an $\alpha >0$, we introduce the weight function
\begin{equation}
 \label{8}
 w_\alpha (x) = \exp(- \alpha |x|), \qquad   x \in \mathbb{R}^d.
\end{equation}
For $x\in \gamma$ and $\theta >0$, we then consider, cf. Eqs. (4) and (5) in \cite{KKP},
\begin{gather}
 \label{9}
a_\gamma(\alpha, \theta):= \sum_{\{x,y\} \in \varepsilon_\gamma} [w_\alpha (x) + w_\alpha (y)] [n_\gamma(x) n_\gamma(y)]^\theta,\\[.1cm]
b_\gamma(\alpha) := \sum_{x\in \gamma} w_\alpha (x), \nonumber
\end{gather}
and
\begin{gather*}
 A_{1,a} = \{ \gamma \in \Gamma: \forall \alpha > 0 \ \forall \theta >0\  \ a_\gamma (\alpha , \theta) < + \infty\} \\[.1cm]
 A_{1,b} = \{ \gamma \in \Gamma: \forall \alpha > 0 \ \ b_\gamma (\alpha) < +\infty\} \nonumber \\[.1cm]
 A_1 = A_{1,a} \cap A_{1,b}. \nonumber
  \end{gather*}
\begin{proposition}
  \label{1pn}
For all $\lambda>0$, it follows that $A_1 \in \mathcal{B}(\Gamma)$ and $\pi_\lambda (A_1) =1$.
\end{proposition}
\begin{proof}
For each $\gamma$, by (\ref{9}) we have that $b_\gamma (\alpha) \leq b_\gamma (\alpha')$ whenever $\alpha > \alpha'$.
Likewise, $a_\gamma(\alpha, \theta)$ is decreasing in $\alpha$ and increasing in $\theta$.
Then
\begin{gather*}
A_{1,a} = \bigcap_{\alpha , \theta \in \mathbb{Q}_{+}} \{ \gamma \in \Gamma:  a_\gamma (\alpha , \theta) < + \infty\}\in \mathcal{B}(\Gamma), \\[.1cm]
A_{1,b} = \bigcap_{\alpha  \in \mathbb{Q}_{+}} \{ \gamma \in \Gamma:  b_\gamma (\alpha ) < + \infty\} \in \mathcal{B}(\Gamma), \nonumber
\end{gather*}
where $\mathbb{Q}_{+}$ is the set of positive rational numbers. Hence,  it is enough to obtain the $\pi_\lambda$-a.s. finiteness of $a_\gamma (\alpha, \theta)$ and $b_\gamma(\alpha)$ for fixed
$\alpha\in \mathbb{Q}_{+}$ and $\theta\in \mathbb{Q}_{+}$.

By the definition of the Poisson measure $\pi_\lambda$, for
each $n\in \mathbb{N}$ and any measurable and symmetric function
$f:\mathbb{R}^d \times \Gamma \to \mathbb{R}_+$, we
have that
\begin{eqnarray}
  \label{10}
 & & \int_\Gamma \left(\sum_{x\in \gamma} f(x , \gamma \setminus x)  \right) \pi_\lambda (d \gamma)\\[.2cm]& &
  \qquad \qquad  \qquad   = \lambda \int_{\Gamma} \left(\int_{\mathbb{R}^d} f(x ,\gamma) dx \right) \pi_\lambda(d \gamma)\nonumber
\end{eqnarray}
-- the Mecke identity.  Then
\begin{eqnarray*}
 \int_{\Gamma} b_\gamma (\alpha) \pi_\lambda (d \gamma) =
 \int_\Gamma \left( \sum_{x\in \gamma} w_\alpha (x) \right) \pi_\lambda (d \gamma) = \lambda \int_{\mathbb{R}^d} w_\alpha(x) d x < \infty,
\end{eqnarray*}
which yields
\[
\pi_\lambda (\{ \gamma: b_\gamma (\alpha) + \infty\}) =0.
\]
Next, we rewrite (\ref{9}) in the form
\begin{equation}
  \label{11}
a_\gamma(\alpha, \theta)= \sum_{x\in \gamma} w_\alpha (x) m_\gamma(\theta, x), \qquad  m_\gamma(\theta, x) := \sum_{y: y\sim x}[n_\gamma(x) n_\gamma(y)]^\theta.
\end{equation}
Set $B_r(x) = \{y\in \mathbb{R}^d: |x-y|\leq r\}$,  $r>0$, and let $\mathcal{I}:\mathbb{R}^d \to \{0,1\}$
be the indicator of the ball $B_{2 r_*}(0)$. Clearly,
\[
\max_{y\in B_{r_*} (x)\cap\gamma} n_\gamma(y) \leq \sum_{z\in \gamma} \mathcal{I}(z-x).
\]
Applying this in (\ref{11}) we get
\begin{equation*}
 m_\gamma(\theta, x) \leq  [n_\gamma(x)]^{\theta +1} \max_{y: y \sim x} [n_\gamma(y)]^\theta \leq \left(\sum_{y\in \gamma\setminus x}  \mathcal{I}(y-x)\right)^{2\theta+1}.
\end{equation*}
By the latter and (\ref{10}), as well as by the translation
invariance of $\pi_\lambda$, we then obtain from (\ref{11})
\begin{eqnarray*}
\int_{\Gamma} a_\gamma(\alpha, \theta)\pi_\lambda (d \gamma) & \leq &  \int_\Gamma \sum_{x\in \gamma}
 w_\alpha (x) \left(\sum_{y \in \gamma\setminus x}\mathcal{I}(y-x) \right)^{2\theta +1} \pi_\lambda (d\gamma)\\
  & = &  \int_{\mathbb{R}^d} w_\alpha (x) \left\{\int_\Gamma \left(\sum_{y\in \gamma }
    \mathcal{I}(y-x) \right)^{2\theta+1}\pi_\lambda (d\gamma)\right\}  dx\\[.2cm]
& = &  \left(\int_{\mathbb{R}^d} w_\alpha (x) dx \right) \cdot
\int_\Gamma \left(\sum_{y\in \gamma }  \mathcal{I}(y)
\right)^{2\theta+1}\pi_\lambda (d\gamma)\\ & = &  \ell_{2\theta+1}
\left( \lambda V (2r^*) \right)  \cdot\int_{\mathbb{R}^d} w_\alpha (x) dx
 < \infty.
\end{eqnarray*}
Here
\[
\ell_{\vartheta} (\varkappa) := e^{-\varkappa}\sum_{k=1}^\infty k^{\vartheta} \varkappa^k/k!,  \qquad \vartheta, \varkappa > 0,
\]
and $V (2r^*)= \int_{\mathbb{R}^d} \mathcal{I}(x) dx$ is the volume of the
ball $B_{2 r_*}(0)$. The  latter estimate leads to the conclusion $$\pi_\lambda ( \{ \gamma:
a_\gamma(\alpha, \theta) = +\infty\}) =0,$$ which completes the proof.
\end{proof}

\subsection{The Gibbs specification}

As mentioned in the Introduction, the a priori distribution of the spin of a `particle' is determined by a finite symmetric measure $\chi$ on $\mathbb{R}$. We assume that, for some $u>2$ and $\varkappa >0$, the following holds
\begin{equation}
  \label{14}
\int_{\mathbb{R}} \exp\left(\varkappa |t|^u \right) \chi(d t)
<\infty.
\end{equation}
For
a fixed $\gamma$, let $\mathbb{R}^\gamma$ stand for the space of all maps $\sigma: \gamma \to \mathbb{R}$. We equip it with
the topology of point-wise convergence and the corresponding Borel $\sigma$-field $\mathcal{B}(\mathbb{R}^\gamma)$.
Let also $\mathcal{P}(\mathbb{R}^\gamma)$ denote the set of all probability measures on $(\mathbb{R}^\gamma, \mathcal{B}(\mathbb{R}^\gamma))$.
For $\Lambda\in \mathcal{B} (\mathbb{R}^d)$ and $\gamma\in \Gamma$, we set
$\gamma_\Lambda = \gamma\cap \Lambda$ and denote
by $\sigma_\Lambda$ the restriction of $\sigma$ to $\gamma_\Lambda$, i.e.,
$\sigma_\Lambda = (\sigma_x)_{x\in \gamma_\Lambda}$.
For $\sigma, \bar{\sigma} \in \mathbb{R}^\gamma$ and $\Lambda \in \mathcal{B}(\mathbb{R}^d)$, by
$\sigma_\Lambda \times \bar{\sigma}_{\Lambda^c}$ we denote the element $\sigma' \in \mathbb{R}^\gamma$ such that
$\sigma'_\Lambda = \sigma_\Lambda$ and
$\sigma'_{\Lambda^c} = \bar{\sigma}_{\Lambda^c}$.

The ferromagnet that we study is characterized by a ferromagnetic spin-spin interaction,
which for fixed $\Lambda\in \mathcal{B}_{\rm b} (\mathbb{R}^d)$ and $\gamma\in \Gamma$ is described by the following relative energy
functional
\begin{eqnarray}
  \label{C2}
  & & - E^{\gamma}_\Lambda (\sigma_\Lambda | \bar{\sigma}_{\Lambda^c})\\[.2cm] & & \quad =
\sum_{\{x,y\} \subset \gamma_\Lambda}
\phi(|x-y|)\sigma_x \sigma_y + \sum_{ x\in \gamma_\Lambda,  y \in \gamma_{\Lambda^c}}
  \phi(|x-y|) \sigma_x \bar{
 \sigma}_y. \nonumber
\end{eqnarray}
Here  $\phi:\mathbb{R}_{+} \to \mathbb{R}_{+}$ is a measurable and bounded
function such that, $\phi(r) \geq \phi_* >0$ for $r\in [0,r_*]$ and $\phi(r) =0$ for all $r>r_*$.

For $\Lambda \in \mathcal{B}_{\rm b}(\mathbb{R}^d)$, by $\mathcal{B}_\Lambda(\mathbb{R}^{\gamma})$ we denote the smallest $\sigma$-subfield of
$\mathcal{B}(\mathbb{R}^{\gamma})$ which contains all sets $A=\{ \sigma \in \mathbb{R}^\gamma: \sigma_\Lambda \in A^0\}$,
$A^0 \in \mathcal{B}(\mathbb{R}^{\gamma_\Lambda})$, where  $\mathcal{B}(\mathbb{R}^{\gamma_\Lambda})$
is the corresponding Borel $\sigma$-field. The algebra of {\it local sets} is
\begin{equation}
 \label{loc}
\mathcal{B}_{\rm loc}  (\mathbb{R}^{\gamma}) := \bigcup_{\Lambda \in \mathcal{B}_{\rm b}(\mathbb{R}^d)}\mathcal{B}_\Lambda(\mathbb{R}^{\gamma}).
 \end{equation}
We will use the following topology on $\mathcal{P}(\mathbb{R}^\gamma)$, see \cite[Definition 4.2, page 59]{G}.
\begin{definition}
  \label{Gedf}
The topology of local set convergence ($\mathfrak{L}$-topology for short) is the weakest topology on $\mathcal{P}(\mathbb{R}^\gamma)$ that makes the evaluation maps
$\mu \mapsto \mu(A)$ continuous for all $A\in \mathcal{B}_{\rm loc}  (\mathbb{R}^{\gamma})$.
\end{definition}
For $\Lambda \in \mathcal{B}_{\rm b}(\mathbb{R}^d)$ and $\bar{\sigma}\in \mathbb{R}^\gamma$, we define
\begin{eqnarray}
  \label{5}
\Pi^\gamma_\Lambda (A|\bar{\sigma}) = \frac{1}{Z^\gamma_\Lambda
(\bar{\sigma})}\int_{\mathbb{R}^{\gamma_\Lambda}} \mathbb{I}_A (
\sigma_\Lambda \times \bar{\sigma}_{\Lambda^c}) \exp(- \beta
E^\gamma_\Lambda (\sigma_\Lambda | \bar{\sigma}_{\Lambda^c}))
\chi_{\Lambda} (d\sigma_\Lambda),
\end{eqnarray}
where $\mathbb{I}_A $ is the indicator of $A\in
\mathcal{B}(\mathbb{R}^\gamma)$, $E^\gamma_\Lambda$ is as in
(\ref{C2}), and
\begin{gather}
  \label{6}
 \chi_{\Lambda} (d\sigma_\Lambda) := \bigotimes_{x\in \gamma_\Lambda} \chi (d \sigma_x),\\
 Z^\gamma_\Lambda (\bar{\sigma}):= \int_{\mathbb{R}^{|\Lambda|}}\exp(- \beta E^\gamma_\Lambda (\sigma_\Lambda | \bar{\sigma}_{\Lambda^c})) \chi_{\Lambda} (d\sigma_\Lambda). \nonumber
\end{gather}
Thus, for each  $A\in \mathcal{B}(\mathbb{R}^\gamma)$, $\Pi^\gamma_\Lambda
(A|\cdot)$ is $\mathcal{B}(\mathbb{R}^\gamma)$-measurable, and, for
each $\bar{\sigma}\in \mathbb{R}^\gamma$, $\Pi^\gamma_\Lambda (\cdot|
\bar{\sigma})$ is a probability measure on
$(\mathbb{R}^\gamma,\mathcal{B}(\mathbb{R}^\gamma))$. The collection
of probability kernels $\{\Pi^\gamma_\Lambda: \Lambda \in \mathcal{B}_{\rm
b} (\mathbb{R}^d)\}$ is called the {\it Gibbs specification} of the
model we consider, see \cite[Chapter 2]{G}. It enjoys the
consistency property
\begin{equation*}
  \int_{\mathbb{R}^\gamma} \Pi^\gamma_{\Lambda_1} (A|\sigma) \Pi^\gamma_{\Lambda_2} (d \sigma|\bar{\sigma}) = \Pi^\gamma_{\Lambda_2} (A|\bar{\sigma}),
\end{equation*}
which holds for all $A\in \mathcal{B}(\mathbb{R}^\gamma)$,
$\bar{\sigma}\in \mathbb{R}^\gamma$, and all $\Lambda_1 , \Lambda_2
\in \mathcal{B}_{\rm b} (\mathbb{R}^d)$ such that $\Lambda_1 \subset
\Lambda_2$.
\begin{definition}
  \label{1df}
A probability measure $\mu$ on
$(\mathbb{R}^\gamma,\mathcal{B}(\mathbb{R}^\gamma))$ is said to be a
\textit{quenched Gibbs measure} (for a fixed $\gamma$) if it
satisfies the Dobrushin-Lanford-Ruelle equation
\[
\mu(A)= \int_{\mathbb{R}^\gamma} \Pi^\gamma_{\Lambda} (A|\sigma)\mu(d \sigma), \qquad {\rm for} \ \ {\rm all} \ \ A\in \mathcal{B}(\mathbb{R}^\gamma).
\]
The set of all such measures is denoted by $\mathcal{G}(\beta|\gamma )$.
\end{definition}
In modern equilibrium statistical mechanics, the notion of
thermodynamic phase of a system of bounded spins living on a fixed
graph like $\mathbb{Z}^d$ is attributed to the extreme elements of the
set of corresponding Gibbs measures, see, e.g., \cite[Chapter
7]{G} or \cite[Chapter III]{Simon}. However, for unbounded spins,
not all extreme Gibbs measures may have physical meaning. It is
believed that the measures corresponding to observed thermodynamic
states should be supported on spin configurations with `tempered
growth' see  \cite{LP,Bell,Park} or a more recent development in
\cite{AKKR,KP} and \cite[Chapter 3]{mon}. In this approach, only
tempered Gibbs measures are taken into account, and hence a phase
transition is related to the existence of multiple tempered Gibbs
measures\footnote{Note that the theory of quantum stabilization
and phase transitions in quantum anharmonic crystals developed in
\cite{AKKR,KaKP,KP} and \cite[Chapter 6]{mon} with the use of
tempered Gibbs measures is consistent with the corresponding
phenomena observed experimentally.}. We take this approach and
study quenched Gibbs measures introduced in Definition \ref{1df}
with a priori prescribed support properties. We call them {\it
tempered Gibbs states}.
Thus, for an $\alpha
>0$, we define
\begin{equation}
\label{18} \Sigma(\alpha):= \bigg{\{} \sigma \in
\mathbb{R}^\gamma:  \sum_{x\in \gamma} |\sigma_x|^2 w_\alpha (x)  <
\infty \bigg{\}},
\end{equation}
where the weights $w_\alpha$ are as in (\ref{8}).
For
each fixed $\gamma$, $\Sigma(\alpha)$ is a Borel subset of $\mathbb{R}^\gamma$ and its elements are called tempered configurations.
Then
\begin{equation}
  \label{t}
\mathcal{G}_{\rm t} (\beta|\gamma ):= \{ \mu \in \mathcal{G} (\beta|\gamma ): \mu (\Sigma(\alpha)) = 1\}
\end{equation}
is the set of tempered Gibbs states.
\begin{theorem}
  \label{1tm}
Let the single-spin measure $\chi$ be such that (\ref{14}) holds and $A_1$ be as in Proposition \ref{1pn}.
Then, for all $\gamma\in A_1$ and all $\beta>0$, the set of Gibbs states $\mathcal{G}_{\rm t}(\beta|\gamma )$ is nonempty.
Moreover, for each $\gamma\in A_1$, and for each positive $\vartheta$ and
$\alpha$, there exists a finite $C_\gamma(\vartheta,
\alpha)>0$ such that the estimate
\begin{equation}
  \label{15}
\int_{\mathbb{R}^\gamma} \exp\left( \vartheta \sum_{x\in \gamma}
|\sigma_x|^2 w_\alpha (x) \right) \mu (d \sigma) \leq C_\gamma (\vartheta,
\alpha)
\end{equation}
holds uniformly for all $\mu \in \mathcal{G}_{\rm
t}(\beta|\gamma )$.
\end{theorem}
\begin{proof}
The proof of all the statements of the theorem follows by Theorem 1 of \cite{KKP} since all the
conditions of that theorem are satisfied in view of
Proposition \ref{1pn} and the assumed properties of $\chi$.
\end{proof}
Let us make some comments. The
existence of Gibbs measures follows from the relative
compactness of the family $\{\Pi^\gamma_\Lambda
(\cdot|\bar{\sigma}): \Lambda \in \mathcal{B}_{\rm b}
(\mathbb{R}^d)\}$ in the $\mathfrak{L}$-topology, see Definition \ref{Gedf}, for at least some $\bar{\sigma}\in
\Sigma(\alpha)$. A typical choice of $\bar{\sigma}$,
for which the compactness is proven, is $\bar{\sigma}_x = s\in
\mathbb{R}$ for all $x\in \gamma$. Note that such $\bar{\sigma}$ is
tempered, see (\ref{18}). Then the accumulation points of the family
$\{\Pi^\gamma_\Lambda (\cdot|\bar{\sigma}): \Lambda \in
\mathcal{B}_{\rm b} (\mathbb{R}^d)\}$ are shown to obey the
Dobrushin-Lanford-Ruelle equation and to satisfy the estimate in
(\ref{15}), in which the constant $C_\gamma(\vartheta, \alpha)$ can be
expressed explicitly in terms of the weights as in (\ref{8}) and the
parameters defined in (\ref{9}), cf. Proposition \ref{1pn} above. Thus, the accumulation points are
tempered measures, and hence belong to $\mathcal{G}_{\rm t} (\beta|\gamma )$.
Note that $\mathcal{G}_{\rm t} (\beta|\gamma )$ is nonempty for all
those $\gamma$, for which both $a_\gamma (\alpha, \theta)$ and $b_\gamma(\alpha)$
are finite. By similar arguments, one can show that
$\mathcal{G}_{\rm t}(\beta|\gamma )$ is compact in the $\mathfrak{L}$-topology.
Finally, let us mention that $u>2$ in (\ref{14}) and $\theta$ in Proposition \ref{1pn} should be such that
$\theta (u-2) > 2$. Under this condition the sufficiently fast decay of the tail of $\chi$
compensates destabilizing effect of the property (\ref{A6}) of the underlying graph, see \cite{KKP} for more detail. Since
Proposition \ref{1pn} holds for all $\theta>0$, then we just assume that $u>2$.

A sequence $\{\Lambda_n\}_{n\in \mathbb{N}}\subset \mathcal{B}_{\rm
b}(\mathbb{R}^d)$ is called {\it cofinal} if
$\Lambda_{n} \subset \Lambda_{n+1}$, $n\in \mathbb{N}$, and each
$\Lambda  \in \mathcal{B}_{\rm b}(\mathbb{R}^d)$ is contained in a
certain $\Lambda_n$. Given $\bar{\sigma}\in \Sigma(\alpha)$, the relative  compactness of the family
$\{\Pi^\gamma_\Lambda (\cdot|\bar{\sigma}): \Lambda \in \mathcal{B}_{\rm b}
(\mathbb{R}^d)\}$ yields that there exists a
cofinal sequence $\{\Lambda_n\}_{n\in \mathbb{N}}$ such that the
sequence $\{\Pi^\gamma_{\Lambda_n} (\cdot|\bar{\sigma})\}_{n\in
\mathbb{N}}$  converges in the $\mathfrak{L}$-topology  to a certain
element of $\mathcal{G}_{\rm t}(\beta|\gamma)$.
For $a> 0$, by
\begin{equation}
  \label{mu-a}
  \mu^{\pm
a}\in \mathcal{G}_{\rm t}(\beta|\gamma)
\end{equation}
we denote limiting elements of $\mathcal{G}_{\rm t}(\beta|\gamma)$
which correspond to $\bar{\sigma}_x= \pm a$ for all  $x \in \gamma$.  Each such $\mu^{\pm a}$
depends on the sequence $\{\Lambda_n\}_{n\in \mathbb{N}}$ along which
it has been attained. Note that only limiting Gibbs measures can approximate thermodynamic states of large finite systems, see \cite[Section 7.1]{G}.

Now we turn to the single-spin measure $\chi$. If it has compact
support, as was the case in \cite{Romano}, then (\ref{14}) clearly
holds for any $u$ and $\varkappa$. The most known example of such
$\chi$ is
\begin{equation}
  \label{16}
 \chi(d t) = [\delta_{-1} ( dt) + \delta_{+1} ( dt)]/2,
\end{equation}
which corresponds to an Ising magnet. Here  $\delta_s$ is the Dirac
measure concentrated at $s\in\mathbb{R}$. We reserve a special notation
$\mathcal{G}^{ \mathrm{Ising}}(\beta|\gamma)$ for the set of all
corresponding Gibbs measures. By
$\nu^{\pm}\in \mathcal{G}^{ \mathrm{Ising}}(\beta|\gamma)$, we denote
the limiting Gibbs measures as in  (\ref{mu-a}) with $a =1$. In this
case, however,  $\nu^{\pm}$ are independent
of the sequences $\{\Lambda_n\}_{n\in \mathbb{N}}$ along which they
were attained. This holds because, for each $x$ and $\nu \in
\mathcal{G}^{ \mathrm{Ising}}(\beta|\gamma)$,
\[
\int_{\mathbb{R}^\gamma} \sigma_x \nu^{-} (d \sigma) \leq
\int_{\mathbb{R}^\gamma} \sigma_x \nu(d \sigma) \leq
\int_{\mathbb{R}^\gamma} \sigma_x \nu^{+} (d \sigma).
\]
That is, $\nu^{+}$ and $\nu^{-}$ are the maximum  and minimum
elements of $\mathcal{G}^{ \mathrm{Ising}}(\beta|\gamma)$,
respectively, cf. \cite[Theorem 3.8]{KP}.

In the case of `unbounded' spins, a natural choice of the
single-spin  measure is
\begin{equation*}
  \chi(d t) = \exp\left( - V(t)\right) dt,
\end{equation*}
where $V: \mathbb{R} \to \mathbb{R}$ is a measurable even function
such that: (a) the set $\{t\in \mathbb{R}: V(t) < +\infty\}$ is of
positive Lebesgue measure; (b) $V(t)$ increases at infinity as
$|t|^{u + \epsilon}$ with some $\epsilon>0$ and $u$ being as in
(\ref{14}). This includes the case where $V$ is a polynomial of even
degree at least $4$ with positive leading coefficient, cf.
\cite{KKP,KKP1,KP,LP}.

\subsection{The question of uniqueness}
\label{2.4ss}

Once the existence of Gibbs states has been established, the next
the problem of their uniqueness/nonuniqueness arises.
Then a phase transition is understood as the possibility to
pass from the uniqueness to nonuniqueness by changing the relevant parameters of the model.
Thus, prior to proving non-uniqueness of $\mu \in \mathcal{G}_{\rm t}(\beta|\gamma)$, which holds for $\pi_\lambda$-almost all $\gamma$ whenever $\beta$ and
$\lambda$ are large enough, see Theorem \ref{2tm} below, we address the question
of whether the same uniqueness does actually hold for some values of these parameters. For small enough $\lambda$,   all the connected
components of the graph $(\gamma, \varepsilon_\gamma)$ are
finite, and hence $\mathcal{G}_{\rm t}(\beta|\gamma)$
is a singleton
for all $\beta$. On the other hand, $\mathcal{G}_{\rm t}(\beta|\gamma)$ is a singleton  if and only if, for each $x\in \gamma$,
arbitrary $\bar{\sigma}\in \Sigma(\alpha)$, and any cofinal sequence
$\{\Lambda_n\}_{n \in \mathbb{N}}$, one has
\begin{equation}
  \label{un}
\lim_{n\to +\infty} \int_{\mathbb{R}^\gamma} \sigma_x
\Pi^\gamma_{\Lambda_n}( d\sigma|\bar{\sigma}) = 0.
\end{equation}
This equivalence holds for any symmetric ferromagnet satisfying the bound in (\ref{15}), that can be proven by
standard arguments based on the Strassen theorem, see \cite{KP} for
more detail.
Actually, for models with `unbounded' spins living on an infinite
connected graph with globally unbounded degree,  which by (\ref{A6}) is the case in our situation, there are no tools\footnote{The celebrated Dobrushin
uniqueness technique is not applicable here.} for
proving (\ref{un}). For the Ising
ferromagnet, the uniqueness in question  can be obtained by percolation
arguments, see \cite[Theorem 7.2]{GHM}.

\section{The Phase Transition}
\label{sec4}
\subsection{The statement}
\label{ss4}

Recall that
by a phase transition in the considered
ferromagnet we mean the fact that the set of tempered Gibbs states
$\mathcal{G}_{\rm t}(\beta|\gamma)$ for $\pi_\lambda$-almost all $\gamma$  contains at least two
elements if $\beta$ and $\lambda$ are big enough. It is equivalent to the appearance of a nonzero
magnetization in states $\mu^{\pm a}\in \mathcal{G}_{\rm t}(\beta|\gamma)$, cf. \cite[Chapter 19]{G} and (\ref{un}).

Let us note that there is no interaction between spins in
different connected components of the underlying graph. Then for a phase transition to occur it is
necessary that the graph $(\gamma, \varepsilon_\gamma)$ possess an
infinite connected component, which holds for $\pi_\lambda$-almost all $\gamma$ whenever $\lambda > \lambda_*$, see \cite{MR,Penrose} and also \cite[Corollary 3.7]{GH} and \cite[Theorem 3.1]{GeL}. For $\lambda <\lambda_{\ast}$, we have no infinite
connected component of $(\gamma, \varepsilon_\gamma)$ and thus $\left\vert
\mathcal{G}(\beta|\gamma)\right\vert =1$ for all $\beta$ and $\pi_\lambda$-almost all $\gamma$. In order
to obtain a sufficient condition for a phase transition to  occur,
we will explore the well-known relationship  between the Bernoulli
bond percolation on the fixed sample graph $(\gamma, \varepsilon_\gamma)$, established in Propositions \ref{GHpn} and \ref{GHpn1},
and the existence of multiple Gibbs states in the corresponding
Ising model, established in \cite{H}. Our goal is to prove the following result.
\begin{theorem}
  \label{2tm}
Let the measure $\chi$ be as in Theorem \ref{1tm} and such that
$\chi(\{0\}) < \chi (\mathbb{R})$. Assume also that the intensity
$\lambda $ of the underlying Poisson point process satisfies the
condition $\lambda >\lambda _{\ast }$, and thus the typical graph $(\gamma, \varepsilon_\gamma)$ has an infinite connected component. Then there exists a
constant $\beta^*>0$ such that, for $\beta > \beta^* $ and $\pi_\lambda$-almost all $\gamma$,
the sets $\mathcal{G}_{\rm t}(\beta|\gamma)$ contain at least two elements.
\end{theorem}
The proof of this statement is based on the following result, cf. (\ref{mu-a}).
\begin{lemma}
  \label{1lm}
Let the conditions of Theorem \ref{2tm} be satisfied. Then there
exist $a>0$,  and $\beta_*>0$ such
that, for all $\beta > a^2 \beta_*$, all $\mu^{+a}\in \mathcal{G}(\beta|\gamma)$, and some $o \in \gamma$,
the following estimate holds for $\pi_\lambda$-almost all $\gamma$:
\begin{equation}
  \label{20}
\int_{\mathbb{R}^\gamma} \sigma_o \mu^{+a}(d \sigma) > 0.
\end{equation}
\end{lemma}
The proof of this lemma is given in the next subsection. \vskip.2cm
\noindent \textit{Proof of Theorem \ref{2tm}:} Since the integral in
(\ref{20}) is the limit of those in (\ref{un}) with $\bar{\sigma}_x
= a$, then (\ref{20}) contradicts (\ref{un}) and hence implies
non-uniqueness, which ought to hold for $\beta > \beta^*:= a^2 \beta_*$. On the other hand, by the invariance of $\chi$ and
of the interaction in (\ref{C2}) with respect to the transformation
$\sigma \to - \sigma$ and $\bar{\sigma} \to - \bar{\sigma}$, we have
\[
\int_{\mathbb{R}^\gamma} \sigma_o \mu^{+a}(d \sigma) = - \int_{\mathbb{R}^\gamma} \sigma_o \mu^{-a}(d \sigma).
\]
Then (\ref{20}) yields $\mu^{+a} \neq \mu^{-a}$ and hence the
multiplicity in question. Note that $o$ in (\ref{20}) belongs to the
infinite connected component of $(\gamma, \varepsilon_\gamma)$, and the
integral in (\ref{20}) is the mean value of the spin at this vertex
in state $\mu^{+a}$.

\subsection{Proof of Lemma \ref{1lm}}

First, by means of the percolation arguments of \cite{H},  we prove the
lemma for the Ising model. Then we extend the proof to
the general case by comparison inequalities.

Recall that the single-spin measure of the Ising model is given in
(\ref{16}), $\mathcal{G}^{ \mathrm{Ising}}(\beta|\gamma)$
denotes the set of all corresponding Gibbs measures, and $\nu^+ \in
\mathcal{G}^{\rm Ising}(\beta|\gamma)$ is the maximum
Gibbs measure as in (\ref{mu-a}) with $a=1$. We are going to use the
key fact proven in \cite{H}:   the Ising  model with constant
intensities $\phi(|x-y|) = \phi_*>0$ on the edges of an infinite
graph has at least two phases if and only if the graph admits the
Bernoulli bond percolation with critical probability $q_* \in (0,1)$
if $\beta  >  [\log (1+q_*) - \log(1-q_*)]/2\phi_*$. In our case,
for $\pi_\lambda$-almost all $\gamma$, the graph $\gamma$ admits
this percolation and the threshold probability satisfies $q_* \geq
\lambda_*/\lambda$, see Proposition \ref{GHpn1}. Then, for some
$o\in \gamma$, it follows that
\begin{equation}
  \label{I}
\int_{\mathbb{R}^\gamma} \sigma_o \tilde{\nu}^{+} (d \sigma) >0,
\end{equation}
see \cite[Theorem 2.1]{H} and also the proof of Lemma 4.2 therein.
Here $\tilde{\nu}^{+}$ is the corresponding Gibbs measure of the Ising model with  $\phi(|x-y|) = \phi_*>0$.
By the standard GKS inequality, see, e.g., \cite[Subsection 3.4]{H}, we have
\[
\int_{\mathbb{R}^\gamma} \sigma_o \nu^{+} (d \sigma) \geq \int_{\mathbb{R}^\gamma} \sigma_o \tilde{\nu}^{+} (d \sigma),
\]
which together with (\ref{I}) yields the proof in this case.

Now we turn to the general case and
 estimate the integral in (\ref{20}) from below by the
corresponding integral with respect to the maximum Gibbs measure
$\nu^{+}$ of the Ising model with a rescaled interaction intensity.
The proof of the lemma immediately follows from the Wells inequality
used, e.g., in \cite{OS}.
\begin{proposition}[Wells inequality]
  \label{Wpn}
Let $a>0$ be such that
\begin{equation}
  \label{21}
\chi([a \sqrt{2}, +\infty))\geq \chi([0,a]).
\end{equation}
Then, for each $x\in \gamma$ and each $\mu^{+a} \in
\mathcal{G}(\beta|\gamma)$, as well as for $\nu^{+} \in
\mathcal{G}^{\rm Ising}(a^2\beta|\gamma)$, we have that
\begin{equation}
  \label{22}
\int_{\mathbb{R}^\gamma} \sigma_x \mu^{+a} (d\sigma) \geq a \int_{\mathbb{R}^\gamma} \sigma_x \nu^{+} (d\sigma) .
\end{equation}
\end{proposition}
As the original publication \cite{W} is hardly attainable and the
proof in \cite{Bricmont} contains numerical inaccuracies, for the reader convenience
in Section \ref{Asec} we give a short proof of this inequality in
the form suitable for our purposes.

\section{Random Gibbs measures}
\label{Quensec}

Random Gibbs measures of spin systems are supposed to
depend on the random parameters ($\gamma$ in ours case) in a measurable way, see, e.g.,
\cite[Definition 6.2.5]{Bov} and
\cite[Definition 2.3]{KKP1}. At the same time, the measures
introduced in Definition \ref{1df} and (\ref{t}) are defined on
spaces which themselves depend on $\gamma$, and thus one cannot speak of the corresponding measurability in this setting.
In order to settle this
problem, we use the results of \cite{DKKP} where
random Gibbs measures are defined as conditional measures
on spaces of marked configurations.
Here we
outline the main points of this construction and prove the non-uniqueness of such measures, see Theorem \ref{etatm} below.

For our model, the space of marked configurations is
\begin{equation*}
\widehat{\Gamma} = \{ \hat{\gamma}= ( \gamma, \sigma): \gamma \in \Gamma, \ \sigma \in \mathbb{R}^\gamma\},
\end{equation*}
where $\Gamma$ is as in (\ref{C1}).
Thus,  the elements of $\widehat{\Gamma}$ consist of pairs
$(x,\sigma_x)$, and $\hat{\gamma} = \hat{\gamma}'$ implies $\gamma = \gamma'$.
In the sequel, we use  the canonical projection
\begin{equation}
  \label{P1}
\widehat{\Gamma} \ni \hat{\gamma} \mapsto \hat{p}(\hat{\gamma}) =\gamma \in \Gamma,
\end{equation}
and equip $\widehat{\Gamma}$
 with the following topology. Let $f: \mathbb{R}^d \times
\mathbb{R}\to \mathbb{R}$ be a bounded continuous function with support contained in $\Lambda \times \mathbb{R}$ for some
$\Lambda \in \mathcal{B}_{\rm b}(\mathbb{R}^d)$.
The topology in question is the weakest
one  which makes the maps
\begin{equation}
  \label{M}
\widehat{\Gamma} \ni \hat{\gamma} \mapsto \sum_{x \in \hat{p}(\hat{\gamma})} f(x, \sigma_x) \in \mathbb{R}
\end{equation}
continuous for all possible $f: \mathbb{R}^d \times
\mathbb{R}\to \mathbb{R}$ as above. This topology is completely and separably metrizable, see \cite[Section 2]{CG}, and thus
$\widehat{\Gamma}$ is a Polish space. Let $\mathcal{B}(\widehat{\Gamma})$ be the corresponding Borel $\sigma$-field.
From the definition of the topologies of $\Gamma$ and $\widehat{\Gamma}$, it follows that the projection defined in (\ref{P1}) is continuous, and hence, for each $\gamma \in \Gamma$,
we have that
\begin{equation*}
  \hat{p}^{-1}(\{\gamma\})=: \mathbb{R}^\gamma \in  \mathcal{B}(\widehat{\Gamma}).
\end{equation*}
For each fixed $\gamma \in \Gamma$,  $\mathbb{R}^\gamma$ is a Polish space embedded into $\widehat{\Gamma}$, which is a Polish space as well.
By the Kuratowski theorem \cite[page 21]{Pa},
the  latter implies that the Borel $\sigma$-fields  $\mathcal{B}(\mathbb{R}^\gamma)$ and
\begin{equation*}
\mathcal{A}( \mathbb{R}^\gamma) := \{ A \in  \mathcal{B}(\widehat{\Gamma}): A \subset \mathbb{R}^\gamma\}
\end{equation*}
are measurably isomorphic. Thus, any probability measure $\mu$ on $\mathcal{B}(\widehat{\Gamma})$ with the property $\mu(\mathbb{R}^\gamma)=1$
can be redefined as a measure on $\mathcal{B}(\mathbb{R}^\gamma)$.

Let $\mathcal{P}(\widehat{\Gamma})$ be the set of all probability measures on $(\widehat{\Gamma}, \mathcal{B}( \widehat{\Gamma}))$.
We equip it with the topology defined as follows. For a fixed $\Lambda \in \mathcal{B}_{\rm b}(\mathbb{R}^d)$, let $\mathcal{B}_\Lambda (\widehat{\Gamma})$ be the smallest $\sigma$-subfield of $\mathcal{B} (\widehat{\Gamma})$ such that the maps as in (\ref{M})
are $\mathcal{B}_\Lambda (\widehat{\Gamma})$-measurable for all bounded measurable functions $f: \mathbb{R}^d \times \mathbb{R} \to \mathbb{R}$
with support contained in $\Lambda \times \mathbb{R}$. Then we set, cf. (\ref{loc}),
\[
\mathcal{B}_{\rm loc}(\widehat{\Gamma}) = \bigcup_{ \Lambda \in \mathcal{B}_{\rm b}(\mathbb{R}^d)} \mathcal{B}_{\Lambda} (\widehat{\Gamma}).
\]
Now the $\mathfrak{L}$-topology on $\mathcal{P}(\widehat{\Gamma})$ is defined as in Definition \ref{Gedf} by using
$\mathcal{B}_{\rm loc}(\widehat{\Gamma})$ as the algebra of local sets. A map $\widehat{\Gamma} \ni \hat{\gamma} \mapsto \varphi(\hat{\gamma})\in \mathbb{R}$ is called {\it local} if it is $\mathcal{B}_\Lambda (\widehat{\Gamma})$-measurable for some $\Lambda \in \mathcal{B}_{\rm b}(\mathbb{R}^d)$. Local maps
$\Gamma \ni \gamma \mapsto \varphi(\gamma)\in \mathbb{R}$ are defined in the same way.

Let $\bar{\sigma}$ in (\ref{5}) be fixed in such a way that $\sigma_x = s$ for some $s\in \mathbb{R}$ and all $x\in \gamma$. Then, for each $\gamma \in \Gamma$ and $\Lambda \in \mathcal{B}_{\rm b} (\mathbb{R}^d)$ and for such $\bar{\sigma}$, we can define a probability measure on $(\widehat{\Gamma}, \mathcal{B}( \widehat{\Gamma}))$ by setting
\begin{equation}
  \label{P4}
\widehat{\Pi}_\Lambda^s(A|\gamma) = \Pi_\Lambda^\gamma (A \cap\mathbb{R}^\gamma|\bar{\sigma}),
\end{equation}
where $\Pi_\Lambda^\gamma (\cdot|\bar{\sigma})$ is given in (\ref{5}). By the very construction, the map
$\Gamma \ni \gamma \mapsto \widehat{\Pi}_\Lambda^s (A|\gamma)\in \mathbb{R}$ is measurable for each $A\in \mathcal{B}(\widehat{\Gamma})$. Thus,
we can consider
\begin{equation}
 \label{eta}
\hat{\eta}_\Lambda^s (\cdot):= \int_{\Gamma} \widehat{\Pi}_\Lambda^s (\cdot|\gamma)\pi_\lambda (d \gamma)
\end{equation}
which is an element of $\mathcal{P}(\widehat{\Gamma})$, equipped with the $\mathfrak{L}$-topology
defined above. It can be shown, see \cite[Corollary 4.2]{DKKP}, that for our model the
family $\{\hat{\eta}^s_\Lambda\}_{\Lambda \in \mathcal{B}_{\rm b}(\mathbb{R}^d)}$
is relatively compact in $\mathcal{P}(\widehat{\Gamma})$ for each $s\in \mathbb{R}$. Let
$\hat{\eta}^s$ be its accumulation point and $\{\Lambda_n\}_{n\in \mathbb{N}}$ be the cofinal sequence such that
$\hat{\eta}^s_{\Lambda_n} \to \hat{\eta}^s$ as $n \to +\infty$. By  $\hat{\eta}_\Gamma^s$ we denote the projection of
$\hat{\eta}^s$ on $\Gamma$. For a bounded local function $f:\Gamma \to \mathbb{R}$,
$\hat{f}:=f\circ \hat{p}:\widehat{\Gamma} \to \mathbb{R}$ is then also local, and hence
\[
\int_\Gamma f d \pi_\lambda = \int_{\widehat{\Gamma}} \hat{f} d \hat{\eta}^s_{\Lambda_n} \to
\int_{\widehat{\Gamma}} \hat{f} d \hat{\eta}^s = \int_\Gamma f d \hat{\eta}^s_\Gamma, \quad n \to +\infty .
\]
Thus, $\hat{\eta}^s_\Gamma = \pi_\lambda$, and we can disintegrate, cf. (\ref{eta}), and obtain
\begin{equation}
\label{eta1}
\hat{\eta}^s (A) = \int_{\Gamma}\eta^s (A|\gamma) \pi_\lambda (d \gamma), \quad A\in \mathcal{B}(\widehat{\Gamma}),
\end{equation}
where $\eta^s$ is a regular conditional measure such that $\eta^s (A|\gamma) = \eta^s (A\cap \mathbb{R}^\gamma|\gamma)$
for almost all $\gamma$. As in (\ref{P4}), we then redefine
$\eta^s(\cdot|\gamma)$ as a measure on $\mathbb{R}^\gamma$, for which we keep the same notation.  One can prove \cite{DKKP}
that, for almost all $\gamma$, $\eta^s(\cdot|\gamma) \in \mathcal{G}_{\rm t}(\beta|\gamma)$.
Thus,  $\eta^s(\cdot|\gamma)$ is a {\it random Gibbs measure}.

In principle, we could construct such Gibbs measures
without the study performed in Section \ref{sec:2}, just by showing that the
family $\{\hat{\eta}^s_\Lambda\}_{\Lambda \in \mathcal{B}_{\rm b}(\mathbb{R}^d)}$
is relatively compact in $\mathcal{P}(\widehat{\Gamma})$.   However, this way has the following drawbacks: (a) in contrast to those in (\ref{mu-a}), the measures
$\eta^s(\cdot|\gamma)$ need not be limiting and hence cannot approximate thermodynamic states of large finite systems; (b)
there is no control on the sets of $\gamma$, as well as on their dependence on $\lambda$ and $s$, for which $\eta^s(\cdot|\gamma)$ exist, cf. Proposition \ref{1pn}; (c) nothing can be said of the integrability properties of
$\eta^s(\cdot|\gamma)$, cf. (\ref{15}); (d) it is unclear whether we can have $\eta^{+a} \neq \eta^{-a}$. These problems are partly resolved in
the statement below. Recall that each
$\mu^{+a}$ is
a measure on $\mathbb{R}^\gamma$, $\gamma\in A_1$, see Proposition \ref{1pn}, and is attained along a cofinal sequence $\mathcal{L}:= \{\Lambda_n\}_{n\in \mathbb{N}}$, that will be indicated as $\mu^{+a}_{\mathcal{L}}$. For each $\gamma\in A_1$, such measures $\mu^{+a}$ constitute the set of accumulation points of the family
$\{\Pi^\gamma_\Lambda (\cdot|\bar{\sigma}): \Lambda \in \mathcal{B}_{\rm b}
(\mathbb{R}^d)\}$. The meaning of the theorem below is that, for a full $\pi_\lambda$-measure subset $A_1' \subset A_1$, there exists a (measurable) selection $\{\mu^{+a}_{\mathcal{L}(\gamma)}\}_{\gamma \in A_1'}$ such that $\mu^{+a}_{\mathcal{L}(\gamma)} = \eta^{+a}(\cdot|\gamma)$ for all $\gamma \in A_1'$.
 \begin{theorem}
 For arbitrary positive $a$ and $\lambda$, there exists $A_1' \subset A_1$ such that: (i) $\pi_\lambda (A_1') =1$; (ii) for each $\gamma \in A_1'$,
 there exists a cofinal sequence $\mathcal{L}(\gamma)$ such that $\mu^{+a}_{\mathcal{L}(\gamma)} = \eta^{+a}(\cdot|\gamma)$. Therefore, $\eta^{+a}(\cdot |\gamma) \neq \eta^{-a}(\cdot |\gamma)$, and hence quenched random Gibbs measures
 are multiple whenever the conditions
of Theorem \ref{2tm} are satisfied.
 \label{etatm}
\end{theorem}
\begin{proof}
From now on we fix $s = +a  $ and $\{\Lambda_n\}_{n\in \mathbb{N}}$, and hence $\hat{\eta}^s$.
Note that we cannot expect that the assumed convergence $\hat{\eta}^s_{\Lambda_n} \to \hat{\eta}^s$ does imply
that $\Pi^\gamma_{\Lambda_n} (\cdot|\bar{\sigma}) \to \eta^s (\cdot|\gamma)$, which would yield the proof.

Let $\{\Xi_n\}_{n\in \mathbb{N}}\subset \mathcal{B}_{\rm b}(\mathbb{R}^d)$ be a partition of
$\mathbb{R}^d$. Fix dense subsets $\{y^n_k\}_{k\in \mathbb{N}} \subset \Xi_n$. As in the proof of Lemma 2.3 in \cite[page 20]{Kall}, by means of this partition we introduce a linear order of the elements
of each $\gamma\in \Gamma$. If $x, y\in \gamma$ belong to distinct $\Xi_n$, we set $x < y$ whenever $x\in \Xi_n$, $y\in \Xi_m$, and $n<m$.  If both $x$ and $y$ lie in the same $\Xi_n$, let $k$ be the smallest integer such that $|x-y^n_k | \neq |y-y^n_k |$. Then we set
$x<y$ if $|x-y^n_k | <|y-y^n_k |$. Next, we enumerate the elements of $\gamma$ in accordance with the order in such a way that $x_1 < x_2 < \cdots < x_k < \cdots$. This defines an
{\it enumeration} on $\widehat{\Gamma}$, that is,  the map
\begin{equation}
  \label{et}
\hat{\gamma} \mapsto \epsilon(\hat{\gamma}) = \{ (x_1 , \sigma_{x_1}), \dots , (x_k , \sigma_{x_k}), \dots \}
\end{equation}
such that all $(x_k, \sigma_{x_k})$ are distinct and $(x_k, \sigma_{x_k}) \in \hat{\gamma}$. This map is
{\it measurable} in the sense that $\{ \hat{\gamma} : x_k \in \Delta, \ \sigma_{x_k} \in \Sigma\} \in \mathcal{B}(\widehat{\Gamma})$ for
each $k\in \mathbb{N}$, $\Delta \in \mathcal{B}_{\rm b}(\mathbb{R}^d)$ and $\Sigma \in \mathcal{B}(\mathbb{R})$.
The latter fact can be proven by a slight generalization of the proof of Lemma 2.3 in \cite[page 20]{Kall}.
Then, for a fixed $\Delta \in \mathcal{B}_{\rm b}(\mathbb{R}^d)$
and any $k\in \mathbb{N}$,
 $p \in (\mathbb{Q}^d)^{\mathbb{N}}$, and $q\in\mathbb{Q}^{\mathbb{N}}$, the real and imaginary parts of the function
\[
 \hat{\gamma} \mapsto \exp[i (p_k\cdot x_k + q_k\sigma_{x_k})]\in \mathbb{C}
\]
are as that in (\ref{M}), and hence are $\mathcal{B}_\Delta (\widehat{\Gamma})$- measurable. Here $i = \sqrt{-1}$ and $p\cdot x $ stands for the scalar product in $\mathbb{R}^d$.
 Let $\mathcal{D}=\{\Delta_n\}_{n\in \mathbb{N}}\subset \mathcal{B}_{\rm b}(\mathbb{R}^d)$ be a cofinal sequence.
For $\epsilon$ as in (\ref{et}), $\Delta \in \mathcal{D}$, $p \in (\mathbb{Q}^d)^{\mathbb{N}}$, and $q\in\mathbb{Q}^{\mathbb{N}}$, we set
\begin{equation}
 \label{eta4}
 f_{\Delta, p, q} (\hat{\gamma}) = \exp\bigg{(} i \sum_{j=1}^{|\gamma_\Delta|} \left[ p_{k_j}\cdot x_{k_j} + q_{k_j} \sigma_{x_{k_j}} \right] \bigg{)},
\end{equation}
where $\gamma = \hat{p}(\hat{\gamma})$, $(x_k, \sigma_{x_k})$ is the $k$-th element of the sequence $\epsilon(\hat{\gamma})$, and
the sum runs over the set of all those $j$ for which $x_{k_j} \in \Delta$. By construction, each such $f$ is $\mathcal{B}_\Delta(\widehat{\Gamma})$-measurable.
Let $\mathcal{F}$ be the (countable) family of all such functions. It has the following properties: (a) is closed under point-wise multiplication;
(b) separates points of $\widehat{\Gamma}$. The latter means that, for any two distinct
$\hat{\gamma}, \hat{\gamma}' \in \widehat{\Gamma}$, one finds $f\in \mathcal{F}$ such that $f(\hat{\gamma}) \neq f(\hat{\gamma}')$.
By  Fernique's theorem \cite[page 6]{Va}, property (b) implies that the smallest $\sigma$-field of subsets of $\widehat{\Gamma}$ is $\mathcal{B}(\widehat{\Gamma})$. Combining this with property (a) we then obtain, see  \cite[page 149]{Bogach} or the proof of Theorem 1.3.26 in \cite[page 113]{mon},
that $\mathcal{F}$ is a separating class for $\mathcal{P}(\widehat{\Gamma})$.
That is, $\mu,\nu \in \mathcal{P}(\widehat{\Gamma})$ coincide if and only if
\[
 \int f d \mu = \int f d \nu, \qquad {\rm for} \ {\rm all} \ f \in \mathcal{F}.
\]
For a fixed triple $\Delta, p, q$, by the assumed convergence $\hat{\eta}^s_{\Lambda_n} \to \hat{\eta}^s$ we have that
\begin{equation}
 \label{eta5}
 \int_{\widehat{\Gamma}} f_{\Delta, p, q} (\hat{\gamma} ) \hat{\eta}^s_{\Lambda_n} ( d\hat{\gamma})= \int_\Gamma g^{(n)}_{\Delta, p, q} (\gamma) \pi_\lambda ( d \gamma) \to \int_{\widehat{\Gamma}} f_{\Delta, p, q} (\hat{\gamma} )
 \hat{\eta}^s ( d\hat{\gamma}),
\end{equation}
as $n \to +\infty$.  Here, cf. (\ref{eta1}) and (\ref{eta4}),
\begin{eqnarray}
\label{eta6}
\qquad g^{(n)}_{\Delta, p, q} (\gamma)& := & \int_{\widehat{\Gamma}} f_{\Delta, p, q} (\hat{\gamma} ) \widehat{\Pi}^{s}_{\Lambda_n} ( d \hat{\gamma} |\gamma) = \exp\bigg{(} i \sum_{j=1}^{|\gamma_\Delta|}  p_{k_j}\cdot x_{k_j} \bigg{)} h^{(n)}_{\Delta, q} (\gamma),\\[.2cm]
h^{(n)}_{\Delta, q} (\gamma) & := &  \int_{\mathbb{R}^\gamma} \exp\bigg{(} i \sum_{x \in \gamma_\Delta}
q_{k(x) } \sigma_x \bigg{)}\Pi^\gamma_{\Lambda_n} (d \sigma|\bar{\sigma}),\nonumber
\end{eqnarray}
where $k(x)$ is the number of $x \in \gamma= \hat{p}(\hat{\gamma})$ defined by the enumeration (\ref{et}).
Obviously, $|h^{(n)}_{\Delta,q} (\gamma)|\leq 1$ for all $n\in \mathbb{N}$, and $\{ h^{(n)}_{\Delta,  q} \}_{n\in \mathbb{N}} \subset
L^1 (\Gamma, d \pi_\lambda)$. By Koml\'os' theorem (see, e.g., \cite{KKP1}), there exists a subsequence $\{ h^{(n_l)}_{\Delta,  q} \}_{l\in \mathbb{N}} \subset \{ h^{(n)}_{\Delta,  q} \}_{n\in \mathbb{N}}$ such that, for each further subsequence
$\{ h^{(n_{l_m})}_{\Delta,  q} \}_{m\in \mathbb{N}} \subset \{ h^{(n_l)}_{\Delta,  q} \}_{l\in \mathbb{N}}$, one has
\begin{equation}
  \label{eta8}
\frac{1}{M} \sum_{m=1}^M  h^{(n_{l_m})}_{\Delta,  q} (\gamma) \to h_{\Delta,  q} (\gamma), \qquad {\rm for} \ \pi_\lambda-{\rm a.a.} \ \gamma\in \Gamma,
\end{equation}
 where $h_{\Delta,  q}$ is a certain element of $ L^1 (\Gamma, d \pi_\lambda)$. Note that the subsequence $\{n_l\}_{l\in \mathbb{N}}$ depends on the choice of $\Delta,  q$. However, by the diagonal procedure as in \cite{KKP1} one can pick $\{n_{l_m}\}_{m\in \mathbb{N}} \subset \{n_l\}_{l\in \mathbb{N}}$ such that (\ref{eta8}) holds for all $\Delta \in \mathcal{D}$,
and $q\in\mathbb{Q}^{\mathbb{N}}$.
Then by (\ref{eta6}) and (\ref{eta8}) we get
\begin{equation}
  \label{eta9}
\int_{\mathbb{R}^\gamma} \exp\bigg{(} i \sum_{x \in \gamma_\Delta}
q_{k(x) } \sigma_x \bigg{)} P^\gamma_{M} (d \sigma|\bar{\sigma}) \to h_{\Delta,  q} (\gamma), \quad {\rm for} \ \pi_\lambda-{\rm a.a.} \ \gamma\in \Gamma ,
\end{equation}
where $\Lambda^m := \Lambda_{n_{l_m}}$ and
\begin{equation}
\label{et2}
P^\gamma_{M} (d \sigma|\bar{\sigma}) := \frac{1}{M} \sum_{m=1}^M \Pi^\gamma_{\Lambda^m} (d \sigma| \bar{\sigma}), \quad \ \  \bar{\sigma}_x = s = +a, \quad \ x \in \gamma .
\end{equation}
By Proposition \ref{1pn}, for $\gamma \in A_1$ the sequence
$\{\Pi^\gamma_{\Lambda^m} (\cdot| \bar{\sigma})\}_{m\in \mathbb{N}}$ is relatively compact in the $\mathfrak{L}$-topology. Thus,  one can pick  $\mathcal{L}(\gamma)=\{\Lambda^{m_l}(\gamma)\}_{l\in \mathbb{N}} \subset \{\Lambda^m\}_{m\in \mathbb{N}}$, for which, cf. (\ref{mu-a}),
\begin{equation}
\label{et1}
\Pi^\gamma_{\Lambda^{m_l}(\gamma)} (\cdot| \bar{\sigma}) \to \mu^{+a}_{\mathcal{L}(\gamma)}\in \mathcal{G}_{\rm t}(\beta|\gamma), \quad l\to +\infty.
\end{equation}
Note that the dependence of the sequence $\mathcal{L}(\gamma)$ on $\gamma$ can be very irregular
in view of the so called {\it chaotic size dependence}, see \cite{NS} and the discussion in \cite{KKP1}.
Now let $A_1'\subset A_1$ be such that also (\ref{eta9}) holds for all $\Delta$ and $q$. Then, for a fixed $\gamma \in A_1'$,  by (\ref{et2}) and (\ref{et1}) there exists the subsequence $\{ P^\gamma_{M_l} \}_{l\in \mathbb{N}} \subset \{ P^\gamma_{M} \}_{M\in \mathbb{N}}$ such that  $P^\gamma_{M_l} \to \mu^{+a}_{\mathcal{L}(\tau)}$, which yields
\[
h_{\Delta,  q} (\gamma) = \int_{\mathbb{R}^\gamma} \exp\bigg{(} i \sum_{x \in \gamma_\Delta}
q_{k(x) } \sigma_x \bigg{)} \mu^{+a}_{\mathcal{L}(\gamma)}(d \sigma), \quad \gamma \in A_1' .
\]
Since $h_{\Delta,  q} \in L^1(\Gamma ,d \pi_\lambda)$, we can integrate and by (\ref{eta1}), (\ref{eta5}), (\ref{eta6}), and (\ref{eta8}) obtain
\begin{eqnarray*}
& & \int_{\Gamma} \exp\bigg{(} i \sum_{j=1}^{|\gamma_\Delta|}  p_{k_j}\cdot x_{k_j} \bigg{)}  \left[ \int_{\mathbb{R}^\gamma} \exp\bigg{(} i \sum_{x \in \gamma_\Delta}
q_{k(x) } \sigma_x \bigg{)} \eta^{+a}(d \sigma|\gamma)\right] \pi_\lambda (d \gamma)  \\[.2cm]
& & =  \int_{\Gamma} \exp\bigg{(} i \sum_{j=1}^{|\gamma_\Delta|}  p_{k_j}\cdot x_{k_j} \bigg{)}  \left[ \int_{\mathbb{R}^\gamma} \exp\bigg{(} i \sum_{x \in \gamma_\Delta}
q_{k(x) } \sigma_x \bigg{)} \mu^{+a}_{\mathcal{L}(\gamma)}(d \sigma)\right] \pi_\lambda (d \gamma),
\end{eqnarray*}
which holds for all $\Delta \in \mathcal{D}$, $p \in (\mathbb{Q}^d)^{\mathbb{N}}$, and $q\in\mathbb{Q}^{\mathbb{N}}$.
This yields the proof since the integrand functions  constitute  separating classes.
\end{proof}

\section{Appendix: Proof of Proposition \ref{Wpn}}
\label{Asec}

For the general choice of $\chi$, let $\Pi^{\gamma,+a}_\Lambda$ be
defined as in (\ref{5}) with $\bar{\sigma}_x = +a$ for all $x\in
\gamma$. Each $\mu^{+a}$ is the weak limit of
$\{\Pi_{\Lambda_n}^{\gamma,+a}\}_{n\in \mathbb{N}}$ for some cofinal
sequence $\{\Lambda_n\}_{n\in \mathbb{N}}$. In the case of unbounded
spins, this convergence alone does not yet imply the convergence of
the moments of $\Pi_{\Lambda_n}^{\gamma,+a}$ to that on the
left-hand side of (\ref{22}). Then we use the uniform in $n$ bound
as in (\ref{15}), which can also be proven for all
$\Pi_\Lambda^{\gamma,+a}$, and obtain
\begin{equation*}
\int_{\mathbb{R}^\gamma} \sigma_x \Pi_{\Lambda_n}^{\gamma,+a} (d
\sigma) \to  \int_{\mathbb{R}^\gamma} \sigma_x \mu^{+a} (d\sigma),
\qquad n\to +\infty.
\end{equation*}
Since the sequence $\{\Lambda_n\}_{n\in \mathbb{N}}$ is exhausting,
it contains a cofinal subsequence, $\{\Lambda_{n_k}\}_{k\in
\mathbb{N}}$, such that also
\[
\int_{\mathbb{R}^\gamma} \sigma_x \Pi^{\gamma,{\rm
Ising}}_{\Lambda_{n_k}} (d \sigma) \to  \int_{\mathbb{R}^\gamma}
\sigma_x \nu^{+} (d\sigma), \qquad n\to +\infty,
\]
where $\Pi^{\gamma,{\rm Ising}}_{\Lambda_{n_k}}$ is the kernel
(\ref{5}) corresponding to the Ising single-spin measure (\ref{16}),
interaction intensities $a^2 \phi(|x-y|)$, and the choice
$\bar{\sigma}_x = +1$ for all $x\in \gamma$. Thus, the validity of
(\ref{22}) will follow if we prove that, for each $\Lambda$ which
contains $x$, the following holds
\begin{equation}
  \label{24}
\int_{\mathbb{R}^\gamma} \sigma_x \Pi_{\Lambda}^{\gamma,+a} (d
\sigma)  \geq a \int_{\mathbb{R}^\gamma} \sigma_x \Pi^{\gamma,{\rm
Ising}}_{\Lambda} (d \sigma).
\end{equation}
Let $Z^\gamma_\Lambda (a)$ and $Z_\Lambda^{^\gamma,{\rm Ising}} (1)$
be the corresponding normalizing factors defined in (\ref{6}). Then
by (\ref{5}) we have, cf. (\ref{C2}),
\begin{eqnarray}
\label{25} & & \int_{\mathbb{R}^\gamma} \sigma_x
\Pi_{\Lambda}^{\gamma,+a} (d \sigma)  - a \int_{\mathbb{R}^\gamma}
\sigma_x \Pi^{\gamma,{\rm Ising}}_{\Lambda} (d \sigma) =
 \left( Z_\Lambda (a)Z_\Lambda^{^\gamma,{\rm Ising}} (1) \right)^{-1} \\[.2cm]
& & \quad \times \int_{\mathbb{R}^\gamma }\int_{\mathbb{R}^\gamma}
(\sigma_x - a \tilde{\sigma}_x)\exp\bigg{\{}
\beta\sum_{\{x,y\}\subset \gamma_\Lambda}\phi(|y-z|)
[\sigma_y \sigma_z + a^2\tilde{\sigma}_y \tilde{\sigma}_z] \nonumber \\[.2cm]
& & \qquad \qquad \qquad \qquad  + \sum_{y \in
\gamma_\Lambda}[\sigma_y + a\tilde{\sigma}_y ]K_y \bigg{\}}
\bigotimes_{x\in \gamma_\Lambda} \left(\chi( d \sigma_x)\otimes
\chi^{\rm Ising} (d \tilde{\sigma}_x)\right), \nonumber
\end{eqnarray}
where $\chi^{\rm Ising}$ is given in (\ref{16}) and $K_y = \beta a
\sum_{z\in \gamma_{\Lambda^c}: z \sim  y } \phi(|y-z|)\geq 0$.  Then
(\ref{24}) will follow from the positivity of the integral on the
right-hand side of (\ref{25}). Now we rewrite the integrand in
(\ref{25}) in the variables $u^{\pm}_x := (\sigma_x \pm a
\tilde{\sigma}_x)/\sqrt{2}$, and then expand the exponent and write
the integral as the sum of the products over $x\in \gamma_\Lambda$
of `one-site' integrals having the form
\begin{eqnarray}
  \label{26}
& &  C_x \int_{\mathbb{R}^2} (u^+_x)^{m_x} (u^-_x)^{n_x}  \chi( d
\sigma_x)\otimes \chi^{\rm Ising} (d \tilde{\sigma}_x)\\ & &  = C_x
\int_{\mathbb{R}}\left[(\sigma_x + a)^{m_x} (\sigma_x - a)^{n_x} +
(\sigma_x - a)^{m_x} (\sigma_x + a)^{n_x} \right]\chi( d\sigma_x),
\quad C_x \geq 0. \nonumber
\end{eqnarray}
Thus, to prove the statement we have to show that the integral on
the right-hand side of (\ref{26}) is nonnegative for all values of
$m_x, n_x \in \mathbb{N}_0$. By the assumed symmetry of $\chi$, this
integral vanishes if $m_x$ and $n_x$ are of different parity. If
both $m_x$ and $n_x$ are even, then the positivity is immediate.
Thus, it is left to consider the case where $m_x = 2 k+1$ and $n_x =
2l+1$. By the symmetry of $\chi$, it is enough to take $k\geq l$.
Thus, we have to prove the positivity of the following integral
\begin{eqnarray*}
& & \int_{\mathbb{R}}\left[(\sigma + a)^{2k+1} (\sigma- a)^{2l+1} + (\sigma - a)^{2k+1} (\sigma + a)^{2l +1} \right]\chi( d\sigma) \\
& & \quad = 2 \int_{0}^{+\infty}(\sigma^2 - a^2)^{2l+1}\left[
(\sigma + a)^{k-l} + (\sigma + a)^{k-l}\right] \chi(d \sigma).
\end{eqnarray*}
The function $\varphi(\sigma) := (\sigma + a)^{k-l} + (\sigma +
a)^{k-l}$ is increasing on $[ 0, +\infty)$. The integral on the
right-hand side of the latter equality can be written in the form
\begin{gather}
\label{27}
 \int_{0}^{+\infty}(\sigma^2 - a^2)^{2l+1} \varphi(\sigma) \chi( d\sigma) = I_1 (a) + I_2 (a) + I_3 (a),\\
 I_1 (a) := \int_{0}^a (\sigma^2 - a^2)^{2l+1} \varphi(\sigma) \chi( d\sigma) \geq - a^{4l +2}\varphi(a) \chi([0,a]), \nonumber \\
 I_2 (a) := \int_{a}^{a\sqrt{2}} (\sigma^2 - a^2)^{2l+1} \varphi(\sigma) \chi( d\sigma) \geq 0, \nonumber\\
 I_3 (a) := \int_{a\sqrt{2} }^{+\infty } (\sigma^2 - a^2)^{2l+1} \varphi(\sigma) \chi( d\sigma)  \geq a^{4l +2}\varphi(a\sqrt{2}) \chi([a\sqrt{2}, +\infty)) \nonumber
\end{gather}
In view of (\ref{21}), the sum on the right-hand side of (\ref{27})
is nonnegative, which completes the proof.

\paragraph{Acknowledgement}

This work was
financially supported by the DFG through SFB 701: ``Spektrale
Strukturen und Topologische Methoden in der Mathematik" and through
the research project 436 POL 125/0-1. The support from the ZiF
Research Group "Stochastic Dynamics: Mathematical Theory and
Applications" (Universit\"at Bielefeld) was also very helpful.
The authors benefited form discussions with Philippe Blanchard and
Stas Molchanov, for which they are cordially indebted. The authors thank also
both referees for constructive criticism and valuable suggestions.




\end{document}